\documentclass{article}
\usepackage{arxiv}
\usepackage[utf8]{inputenc}
\usepackage[T1]{fontenc}
\usepackage{hyperref}
\usepackage{url}
\usepackage{booktabs}
\usepackage{amsfonts}
\usepackage{nicefrac}
\usepackage{microtype}
\usepackage{lipsum}
\usepackage{amsmath}
\usepackage{algpseudocode}
\DeclareMathOperator*{\argmin}{arg\,min}
\usepackage{amsthm}
\usepackage{graphicx}
\usepackage{algorithm}
\newtheorem{theorem}{Theorem}
\graphicspath{ {./images/} }
\newtheorem{remark}{Remark}
\usepackage{caption}
\usepackage{authblk}
\title{Tensor Star Decomposition and Its Application to Tensor Completion}

\author[1]{\textbf{Wuyang Zhou}}
\author[2]{\textbf{Yu-Bang Zheng}}
\author[3]{\textbf{Qibin Zhao}}
\author[1]{\textbf{Danilo Mandic}}
\affil[1]{Department of Electrical and Electronic Engineering, Imperial College London}
\affil[2]{School of Information Science and Technology, Southwest Jiaotong University}
\affil[3]{Tensor Learning Team, RIKEN Center for Advanced Intelligence Project}

\begin{document}

\textbf{This work has been submitted to the IEEE for possible publication. Copyright may be transferred without notice, after which this version may no longer be accessible.}
\maketitle
\begin{abstract}
A novel tensor decomposition framework, termed Tensor Star (TS) decomposition, is proposed which represents a new type of tensor network decomposition based on tensor contractions. This is achieved by connecting the core tensors in a ring shape, whereby the core tensors act as skip connections between the factor tensors and allow for direct correlation characterisation between any two arbitrary dimensions. Uniquely, this makes it possible to decompose an order-$N$ tensor into $N$ order-$3$ factor tensors $\{\mathcal{G}_{k}\}_{k=1}^{N}$ and $N$ order-$4$ core tensors $\{\mathcal{C}_{k}\}_{k=1}^{N}$, which are arranged in a star shape. Unlike the class of Tensor Train (TT) decompositions, these factor tensors are not directly connected to one another. The so obtained core tensors also enable consecutive factor tensors to have different latent ranks. In this way, the TS decomposition alleviates the "curse of dimensionality" and controls the "curse of ranks", exhibiting a storage complexity which scales linearly with the number of dimensions and as the fourth power of the ranks.
\end{abstract}

\section{Introduction}
Classic higher-order tensor decompositon frameworks include CANDECOMP/PARAFAC (CP) decomposition \cite{cpd} and Tucker decompositon (TKD) \cite{tkd}, which is a higher-order generalisation of the matrix singular value decompositon (SVD). Higher-Order Singular Value Decomposition (HOSVD) \cite{hosvd} is a specific case of Tucker decomposition, which yields an all-orthogonal core tensor and orthonormal factor matrices. \textbf{Despite the increased interpretability, TKD and HOSVD need to store a core tensor of the same order as the original tensor, causing a high storage complexity}. To address this issue, Tensor Network (TN) based decompositions such as Tensor Train (TT) decomposition \cite{ttd} have become popular in recent years due to its super-compression and efficient computational properties. To solve the problem of rank accumulation and sequential multiplication in TT decompositons, the Tensor Ring (TR) decomposition \cite{zhao2016tensor} was proposed, which replaces the two order-$2$ tensors in the TT decomposition by order-$3$ core tensors, and connects both ends. \textbf{However, the lack of a core structure that directly connects any two factor tensors in TT and TR decomposition limits their interpretability and only allows for the correlation characterisation between two consecutive dimensions.} A step-forward was introduced by adding an order-$N$ core tensor to TR decompositon, so called the Tensor Wheel (TW) decomposition  \cite{twd}. Other recent efforts to find a more efficient tensor network representation include the Fully-Connected Tensor Network (FCTN) Decomposition \cite{zheng2021fully}, which decomposes an order-$N$ tensor into $N$ small-sized order-$N$ fully-connected tensors. \textbf{While TW and FCTN decompositions exhibit stronger correlation characterization abilities compared to TT and TR decompositions, they suffer from significantly higher storage complexity, that is, $\mathcal{O}(NIR^3+R^N)$ and $\mathcal{O}(NIR^{N-1})$, respectively.}

\renewcommand{\arraystretch}{1.5}
\begin{table}[t]
\centering
\begin{tabular}{|c|c|}
\hline
\textbf{Model} & \textbf{Storage Complexity} \\
\hline
\hline
Original tensor & $\mathcal{O}(I^N)$ \\\hline
CP decomposition & $\mathcal{O}(NIR)$ \\\hline
Tucker decomposition & $\mathcal{O}(R^N + NIR)$ \\\hline
TT decomposition & $\mathcal{O}(NIR^2)$ \\\hline
TR decomposition & $\mathcal{O}(NIR^2)$ \\\hline
FCTN decomposition & $\mathcal{O}(NIR^{N-1})$ \\\hline
TW decomposition & $\mathcal{O}(NIR^3+R^N)$ \\\hline
Tensor Star decomposition & $\mathcal{O}(NIR^2+NR^4)$ \\\hline

\end{tabular}
\caption{Storage complexities of some established tensor decompositions. Observe that the proposed Tensor Star decomposition is particularly suitable for tensors of high dimension and cores of low rank.}
\label{tab:mathnotation}
\end{table}
\renewcommand{\arraystretch}{1}

\begin{figure}[t]
    \centering
    \includegraphics[width = 1\hsize]{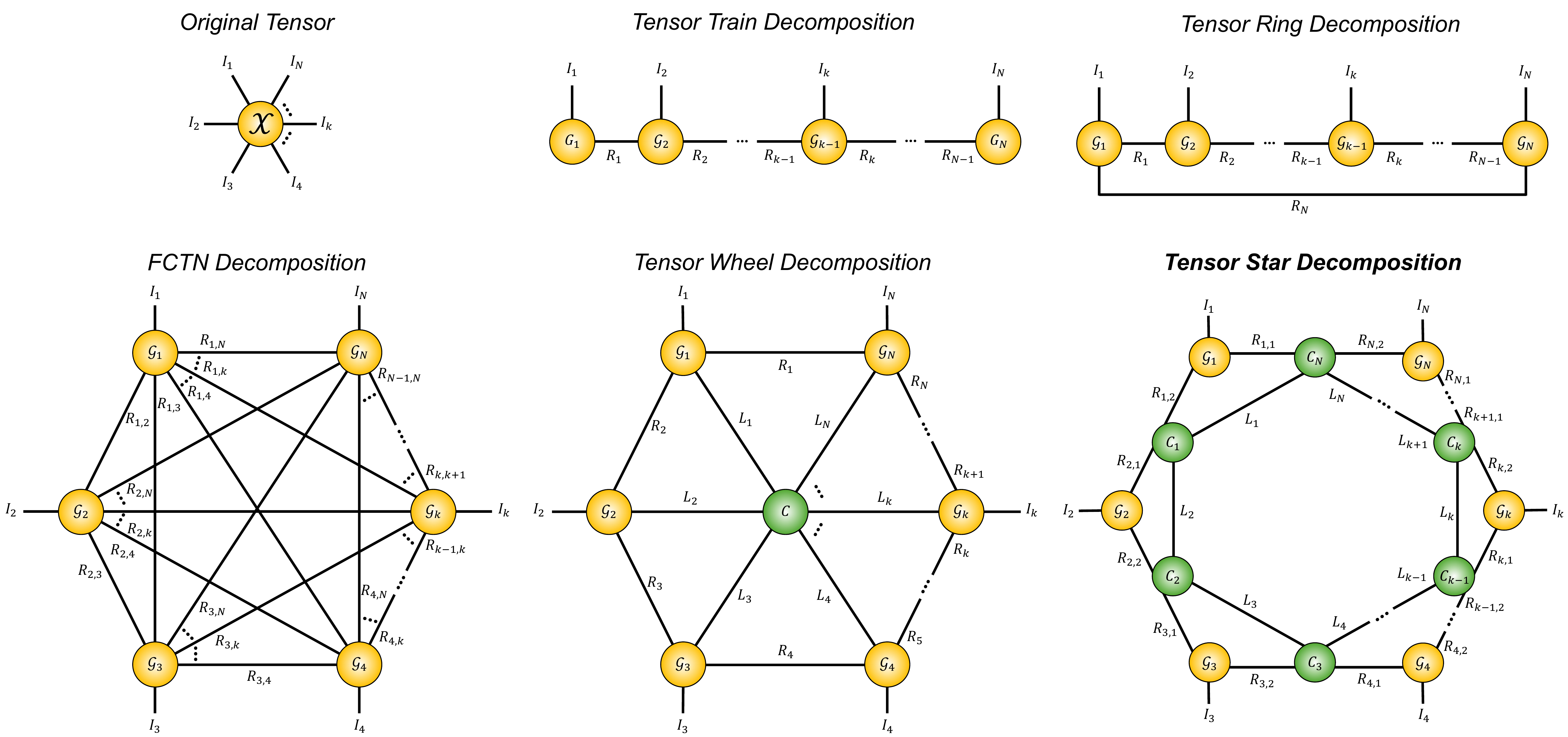}
    \caption{Graphs for some tensor decompositions of an order-$N$ tensor, $\mathcal{X}  \in \mathbb{R}^{I_1 \times I_2 \times \cdots \times I_N}$. The proposed Tensor Star decomposition of order-$N$ is shown at the bottom right.}
    \label{fig:tsd}
\end{figure}

\begin{figure}[t]
    \centering
    \includegraphics[width = 1\hsize]{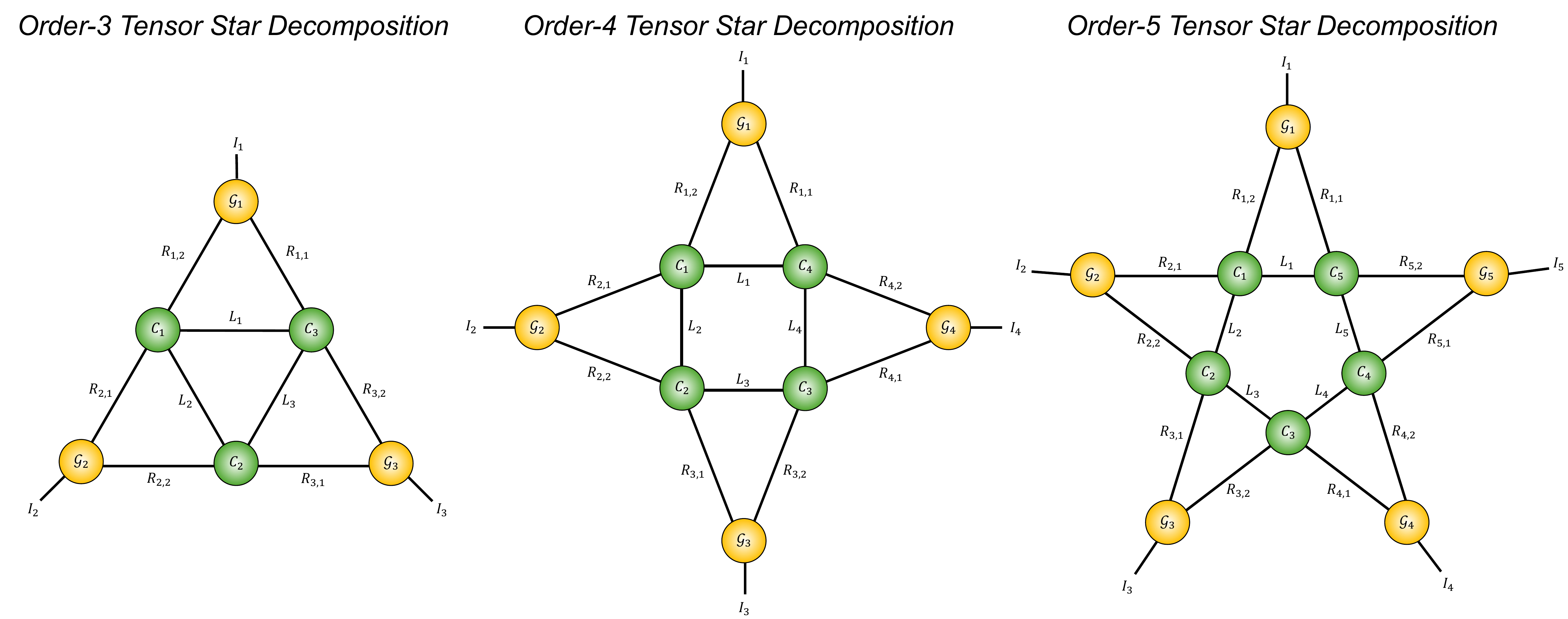}
    \caption{Tensor Star decomposition of order-$3$, order-$4$ and order-$5$}
    \label{fig:tsdn234}
\end{figure}

In this paper, we propose a novel tensor network decomposition, named Tensor Star (TS) Decomposition, which integrates strong correlation characterization abilities with a low storage complexity. To be specific, TS decomposition decomposes an order-$N$ tensor, $\mathcal{X} \in \mathbb{R}^{I_1 \times I_2 \times \cdots \times I_N}$, into $N$ order-$3$ factor tensors $\{\mathcal{G}_{k}\}_{k=1}^{N} \in \mathbb{R}^{R_{k,1}\times I_k \times R_{k,2}}$ and $N$ order-$4$ core tensors  $\{\mathcal{C}_{k}\}_{k=1}^{N} \in \mathbb{R}^{R_{k,2} \times L_k \times L_{k+1}\times R_{k+1,1}}$, which are arranged in a star shape. As shown in Figure \ref{fig:tsd} and Figure \ref{fig:tsdn234}, the core tensors $\{\mathcal{C}_{k}\}_{k=1}^{N}$ isolate the factor tensors from each other and can act as skip connections.  By assuming all latent ranks, $\{R_{k,s}\}_{k=1, s=1}^{N, 2}$ and $\{L_k\}_{k=1}^{N}$, to be of value $R$, and all mode sizes, $\{I_k\}_{k=1}^N$ to be $I$, TS decomposition achieves a storage complexity of $\mathcal{O}(NIR^2+NR^4)$, which scales linearly with the number of dimensions and the fourth power of the ranks. Thus, TS decomposition has a better storage efficiency compared to TKD, HOSVD, FCTN, and TW decomposition. At the same time, TS decomposition allows for a direct correlation characterisation between any two dimensions, in a way similar to TKD, HOSVD, FCTN, and TW decomposition. In this paper, we follow the usual approach to developing new tensor decomposition methods by deriving the definition, properties, two different rank bounds, alternating least squares algorithm, and proximal alternating minimization algorithm for tensor completion of Tensor Star decomposition.

In summary, the advantage of Tensor Star decomposition is threefold:
\begin{enumerate}
    \item The first advantage of TS decomposition is the direct correlation characterisation between any two arbitrary modes while not having a storage complexity that scales exponentially with the number of dimensions.
    \item The second advantage of TS decomposition is in the application of neural network (NN) compression, where one neural network layer is usually tensorized into higher dimensions to have smaller sizes in each mode before applying TN decompositions. Although both TT and TR decompositions enable super-compression, due to the sequential nature of the decomposed tensors, the gradient must go through the sub-tensors which are directly connected to the modes one-by-one, causing the gradient to go very deep implicitly, despite the fact that there maybe only one layer of fully-connected neural network being compressed. This is less of a problem in TKD, HOSVD, FCTN or TW decomposition, but they all have very high storage complexities, making them unsuitable for NN compression. The core innovation of TS decomposition is essentially adding skip connections to the TR topology, providing a highway for the gradient to pass through when compressing neural networks. Thus, TS decomposition is able to pass gradients directly from one mode to another (due to the first advantage), while having a storage complexity that does not scale exponentially with the number of dimensions.
    \item Although not shown directly through the theoretical storage complexity analysis, the third advantage of TS decomposition is that allowing consecutive factor tensors to have different latent ranks help reduce the storage requirements, especially when there is a mode with a very large size (reflected through a high rank) next to a mode with a very small size (reflected through a low rank). If TT or TR decomposition is used, to keep the approximation error low, the shared rank between these two factor tensors needs to be large to accommodate the mode with the large rank. However, in TS decomposition, we can use the core tensor ring connections to pass down the information without using a large latent rank for the factor tensor corresponding to the mode of the small rank.
\end{enumerate}

\section{Preliminaries}

\renewcommand{\arraystretch}{1.5}
\begin{table}[t]
\centering
\begin{tabular}{|c|c|}
\hline
\textbf{Name} & \textbf{Mathematical Notation} \\
\hline
\hline
Scalar & $y$ \\\hline
Vector of size $I_1$ & $\mathbf{y} \in \mathbb{R}^{I_1}$ \\\hline
Matrix of size $I_1 \times I_2$& $\mathbf{Y} \in \mathbb{R}^{I_1 \times I_2}$\\
\hline
Real-valued tensor of order-$N$ &$ \mathcal{Y} \in \mathbb{R}^{I_1 \times I_2 \times \cdots \times I_N}$ \\ \hline
${(i_1, i_2,\ldots, i_N)}$-th entry of $\mathcal{Y}$ & $y_{i_1 i_2 \cdots i_N}$ \\ \hline
Mode permutation according to $\mathbf{v}$& $\vec{\mathcal{Y}}^{\mathbf{v}}$ \\ \hline
Circular Mode permutation according to $\mathbf{n}$& $\vec{\mathcal{Y}}^{\mathbf{n}}$ \\ \hline
Classic mode-n unfolding of a tensor & $\mathcal{Y}_{(n)}$ \\\hline
Classic mode-n folding of a tensor & $\text{Fold}_{(n)}$ \\\hline
Circular mode-n unfolding of a tensor & $\mathcal{Y}_{<n>}$ \\\hline
Circular mode-n folding of a tensor & $\text{Fold}_{<n>}$ \\\hline
Generalised unfolding of a tensor  & $\mathcal{Y}_{[\mathbf{v};d]}$ \\\hline
Generalised row vectorization of a tensor  & $\mathcal{Y}_{[\mathbf{v};0]}$ \\\hline
Generalised column vectorization of a tensor  & $\mathcal{Y}_{[\mathbf{v};N]}$ \\\hline
Generalised folding of a tensor  & $\text{Fold}_{[\mathbf{v};d]}$ \\\hline
Tensor contractions & $\times^{i_1,i_2,\ldots,i_d}_{j_1,j_2,\ldots,j_d}$ \\\hline
Frobenius norm of $\mathcal{Y}$& $\| \mathcal{Y} \|_F$\\\hline
\end{tabular}
\caption{Mathematical notations used}
\label{tab:mathnotation}
\end{table}
\renewcommand{\arraystretch}{1}

Mathematical notations are summarised in Table \ref{tab:mathnotation}. Scalars are denoted by the lowercase letter, $\mathit{y}$. Vectors are denoted by the bold lowercase letter, $\mathbf{y}$. Matrices are denoted by bold capital letters, $\mathbf{Y}$. Tensors are denoted by uppercase calligraphic letters, $\mathcal{Y} \in \mathbb{R}^{I_1 \times I_2 \times \cdots \times I_N}$. The element-wise entry at a location ${(i_1, i_2,\ldots, i_N)}$ of $\mathcal{Y}$ is denoted by $y_{i_1 i_2 \ldots i_N} = \mathcal{Y} ( i_1 , i_2 ,\ldots,i_N )$. The frobenius norm of an order-$N$ tensor, $\mathcal{Y}$, is denoted as $\| \mathcal{Y} \|_F = \sqrt{\sum_{i_1, i_2, \ldots, i_N}^{I_1, I_2, \ldots, I_N} y_{i_1 i_2 \ldots i_N}^2}$.

The circular mode permutation involves circularly shifting the mode indexes in $\mathcal{Y}$. Let $\mathbf{n} = (n_1,n_2,\ldots,n_N)$ be a circular reordering of $(1,2,\ldots,N)$, then the circular mode permutation of $\mathcal{Y}$ is denoted as $\vec{\mathcal{Y}}^{\mathbf{n}} \in \mathbb{R}^{I_{n_1} \times I_{n_2} \times...\times I_{n_N}}$.

Unfolding operations convert tensors into matrices. The classic mode-$n$ unfolding of a tensor $\mathcal{Y}$ is denoted as $\mathbf{Y}_{(n)}(\overline{i_{n}}, \overline{i_1 i_2 \cdots i_{n-1} i_{n+1} \cdots i_{N}}) \in \mathbb{R}^{I_n \times \prod_{k \neq n}^N I_k}$ \cite{tkd}. The circular mode-$n$ unfolding of a tensor $\mathcal{Y}$ represents the classic mode-n unfolding with the circularly re-arranged second-mode. More specifically, the circular mode-n unfolding of $\mathcal{Y}$ is denoted as $\mathbf{Y}_{<n>}(\overline{i_{n}}, \overline{i_{n+1} i_{n+2} \cdots i_{N} i_1 i_2 \cdots i_{k-1}}) \in \mathbb{R}^{I_n \times (\prod_{k = n+1}^N I_k) \cdot (\prod_{j = 1}^{n-1} I_j)}$ \cite{zhao2016tensor}. The inverse operations of the classic mode-$n$ unfolding and the circular mode-$n$ unfolding are denoted as $\text{Fold}_{(n)}$ and $\text{Fold}_{<n>}$, respectively.

It is possible to extend the idea of mode-$n$ unfolding further to yield generalised unfolding, which allows for the permutation of modes through a re-ordering vector $\mathbf{v}$ and unfolds the tensor into two modes, where by the first mode is from the first $d$ permuted modes, and the second mode is from the rest $(N-d)$ permuted modes. For a dimension permutation vector, $\mathbf{v}$, of $(1,2,...,N)$, a generalised tensor unfolding is denoted as (Definition 2 in \cite{zheng2021fully})
\begin{equation}
\mathbf{Y}_{[\mathbf{v};d]}(\overline{i_{v_1} \cdots i_{v_d}}, \overline{i_{v_{d+1}} \cdots i_{v_N}})  \in \mathbb{R}^{ \prod_{k = 1}^d I_{v_k} \times \prod_{j = d+1}^N I_{v_j}} = \vec{\mathcal{Y}}^{\mathbf{v}}(i_{v_1}, i_{v_2}, \ldots, i_{v_d})
\end{equation}

Two special cases of generalised unfolding are the generalised row vectorization and generalised column vectorization \cite{twd}, where the generalised row vectorization is denoted as
\begin{equation}
    \mathbf{y}_{[\mathbf{v};0]}(\overline{i_{v_1} \cdots i_{v_d} i_{v_{d+1}} \cdots i_{v_N} })  \in \mathbb{R}^{ 1 \times (\prod_{k = 1}^d I_{v_k}) \cdot (\prod_{j = d+1}^N I_{v_j})} = \vec{\mathcal{Y}}^{\mathbf{v}}(i_{v_1}, i_{v_2}, \ldots, i_{v_d})
\end{equation}

and the generalised column vectorization as
\begin{equation}
    \mathbf{y}_{[\mathbf{v};N]}(\overline{i_{v_1} \cdots i_{v_d} i_{v_{d+1}} \cdots i_{v_N} })  \in \mathbb{R}^{ (\prod_{k = 1}^d I_{v_k}) \cdot (\prod_{j = d+1}^N I_{v_j}) \times 1} = \vec{\mathcal{Y}}^{\mathbf{v}}(i_{v_1}, i_{v_2}, \ldots, i_{v_d})
\end{equation}

The inverse operations of the generalised unfolding, the generalised row vectorization, and the generalised column vectorization are denoted as $\text{Fold}_{[\mathbf{v};d]}$, $\text{Fold}_{[\mathbf{v};0]}$ and $\text{Fold}_{[\mathbf{v};N]}$, respectively.

All tensor contractions used in this paper are generalised tensor contractions, which are only allowed between modes with common size and are non-commutative and non-associative \cite{zheng2021fully}. They are equivalent to the MATLAB built-in function, $tensorprod(A, B, dimA, dimB)$. Different from the classic tensor contractions, generalised tensor contraction "absorbs" the contracted modes and concatenates the remaining modes of the right-hand tensor to the remaining modes of the left-hand tensor. Assume $\mathcal{A} \in \mathbb{R}^{J_1 \times J_2 \times \cdots \times J_M}$ and $\mathcal{Y} \in \mathbb{R}^{I_1 \times I_2 \times \cdots \times I_N}$ have $d$ common modes $( 0< d \le \min\{M,N\}, \ d \in \mathbb{Z} )$. Define a vector $\mathbf{a_v}$ to be a permutation of $(1,2,\ldots,M)$, and a vector $\mathbf{y_v}$ to be a permutation of $(1,2,\ldots,N)$. If we were to contract all $d$ common modes, we let the first $d$ entries of $\mathbf{a_v}$ and $\mathbf{y_v}$ correspond to the $d$ common modes, i.e., $J_{\mathbf{a_v}(k)} = I_{\mathbf{y_v}(k)}, \ k=1,2,\ldots, d$. The remaining $(M-d)$ entries of $\mathbf{a_v}$ then satisfy $\mathbf{a_v}(k+1) < \mathbf{a_v}(k+2) < \mathbf{a_v}(M)$. Similarly, the remaining $(N-d)$ entries of $\mathbf{y_v}$ satisfy $\mathbf{y_v}(k+1) < \mathbf{y_v}(k+2) < \mathbf{y_v}(N)$. Then, the tensor contraction is expressed as (Definition 3 in \cite{zheng2021fully})

\begin{equation}
    \mathcal{A} \times_{\mathbf{a_v}(1), \mathbf{a_v}(2), \ldots, \mathbf{a_v}(d)}^{\mathbf{y_v}(1), \mathbf{y_v}(2), \ldots, \mathbf{y_v}(d)} \mathcal{Y} = \text{Fold}_{[(1,2,\ldots, M+N-2d);M-d]} (\mathbf{A}^{\mathsf{T}}_{[\mathbf{a_v};d]}  \mathbf{Y}_{[\mathbf{y_v};d]}) \in \mathbb{R}^{J_{d+1} \times J_{d+2} \times \cdots \times J_{M} \times I_{d+1} \times I_{d+2} \times \cdots \times I_{N}}
\end{equation}

\section{Tensor Star (TS) Decomposition }

\subsection{Definition}
As shown in Figure \ref{fig:tsd}, TS decomposition decomposes an order-$N$ tensor $\mathcal{X} \in \mathbb{R}^{I_1 \times I_2 \times \cdots \times I_N}$ into $N$ order-$3$ factor tensors $\mathcal{G}_k \in \mathbb{R}^{R_{k,1}\times I_k \times R_{k,2}}$ and $N$ order-$4$ core tensors $C_k \in \mathbb{R}^{R_{k,2} \times L_k \times L_{k+1}\times R_{k+1,1}}$ for $k=1,2,\ldots,N$. $L_{N+1} = L_1$, $R_{N+1,1} = R_{1,1}$ and $R_{N+1,2} = R_{1,2}$.

Using tensor contractions, we can now express the tensor $\mathcal{X}$ as 
\begin{equation}
    \begin{split}
    \mathcal{X} = \mathcal{G}_1 \times_3  \mathcal{C}_1 \times_5 \mathcal{G}_2 \times_{4,6}^{2,1}  \mathcal{C}_2 \times_{6} \mathcal{G}_3 \times_{5,7}^{2,1}  \mathcal{C}_3 \times_{7} \cdots \times_{k+3} \mathcal{G}_k \times_{k+2,k+4}^{2,1} \mathcal{C}_k \times_{k+4} \cdots \times_{N+3} \mathcal{G}_N \times_{1,3,N+2,N+4}^{4,3,2,1} \mathcal{C}_N
    \end{split}
    \label{equ:tsd_contraction}
\end{equation}
To express (\ref{equ:tsd_contraction}) in a compact form, we define the TS decomposition operator $TSD(\mathcal{G}_1, \mathcal{C}_1, \mathcal{G}_2, \mathcal{C}_2, \ldots, \mathcal{G}_N, \mathcal{C}_N)$ or simply $TSD(\{ \mathcal{G}_k, \mathcal{C}_k\}^N_{k=1})$. In this way, each element in $\mathcal{X}$ can be expressed as
\begin{equation}
    \begin{split}
    \mathcal{X}(i_1,i_2,...,i_N) =
    & \sum_{r_{1,1}=1}^{R_{1,1}}\sum_{r_{1,2}=1}^{R_{1,2}}\sum_{r_{2,1}=1}^{R_{2,1}}\sum_{r_{2,2}=1}^{R_{2,2}}\cdots\sum_{r_{N,1}=1}^{R_{N,1}}\sum_{r_{N,2}=1}^{R_{N,2}}\sum_{l_1=1}^{L_1}\sum_{l_2=1}^{L_2}\cdots \sum_{l_N=1}^{L_N} \\
     & \{ \mathcal{G}_1(r_{1,1},i_1,r_{1,2}) \ \mathcal{C}_1(r_{1,2}, l_1, l_{2}, r_{2,1}) \ \mathcal{G}_2(r_{2,1},i_2,r_{2,2}) 
 \ \mathcal{C}_2(r_{2,2}, l_2, l_{3}, r_{3,1}) \cdots \\
    & \ \ \mathcal{G}_k(r_{k,1},i_k,r_{k,2}) \ \mathcal{C}_k(r_{k,2}, l_k, l_{k+1}, r_{k+1,1}) \cdots \mathcal{G}_N(r_{N,1},i_N,r_{N,2}) \ \mathcal{C}_N(r_{N,2}, l_N, l_{1}, r_{1,1}) \}
    \end{split}
\end{equation}

The TS decomposition factor ranks are defined as $R_{k,s} : 1 \leq k \leq N \ and \  s \in \{1, 2\}$, while the TS decomposition ring ranks are defined as $L_k : 1 \leq k \leq N$.

\subsection{Properties}
Although both TT decomposition and TR decomposition enable super-compression, their approaches to breaking "the curse of dimensionality" impede the ability to capture the intrinsic interactions between non-adjacent dimensions. Any mode interaction can be fully exploited in the FCTN decomposition and TW decomposition. However, they both have storage complexities scaling exponentially with the number of dimensions, $N$. On the otherhand, as shown in Theorem \ref{theorem2} , the proposed TS decomposition is able to utilise the mode interactions between any two arbitrary modes, while having a storage complexity of $\mathcal{O}(NIR^2+NR^4)$, which is significantly smaller than the storage complexity of TW decomposition, $\mathcal{O}(NIR^3+R^N)$, and the storage complexity of FCTN decomposition,  $\mathcal{O}(NIR^{N-1})$.
\begin{figure}[t] 
    \centering
    \includegraphics[width = 0.4\hsize]{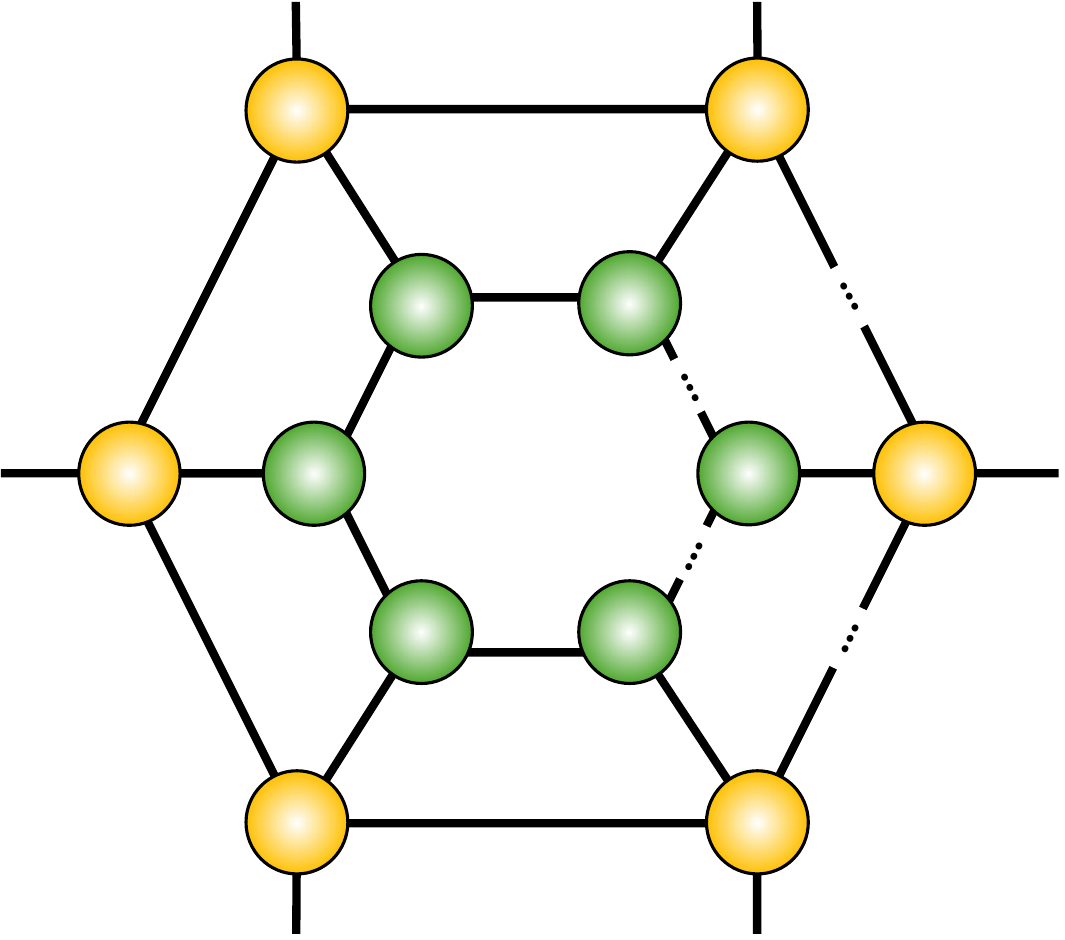}
    \caption{Tensor Wheel (TW) decomposition  with a ring-shaped core}
    \label{fig:twd_ring}
\end{figure}

\begin{remark}
The TS decomposition has similar desirable properties to TW decomposition, while having a storage complexity which does not scale exponentially with the number of dimensions, $N$. As shown in Figure \ref{fig:twd_ring}, a similar framework can be established by decomposing the core tensor in TW decomposition into sub-core tensors connected in a ring topology, which has not yet been formally proposed and studied. Furthermore, TW decomposition with the so decomposed core tensor still has one shared rank between every two consecutive factor tensors.
\end{remark}

\begin{remark}
In TT, TR, and TW decomposition, there must be one rank (connection) shared between every two consecutive factor tensors. For FCTN decomposition, this is true also for any two arbitrary factor tensors. In TS decomposition, for every two consecutive modes, we treat the problem as a decomposition problem of an order-$4$ tensor where only a Tucker-$2$ model is applied. This not only allows every two consecutive factor tensors to have different latent ranks. In this way, there will also be two more modes in the core tensors, which allow the core tensors to be connected in a ring topology.
\end{remark}

\begin{theorem}[\textbf{Circular dimensional permutation invariance}]
Assume that an order-$N$ tensor $\mathcal{X} \in \mathbb{R}^{I_1 \times I_2 \times \cdots \times I_N}$ has a Tensor Star decomposition $TSD(\{ \mathcal{G}_k, \mathcal{C}_k\}^N_{k=1})$. Then, circularly reordering the dimensions of $\mathcal{X}$ does not change its TS decomposition factor tensors and TS decomposition core tensors. Let $\mathbf{n} = (n_1, n_2, \ldots, n_N)$ be a circular dimensional permutation of $(1,2,\ldots,N)$, then $\vec{\mathcal{X}}^{\mathbf{n}}$ can be expressed as
$$\vec{\mathcal{X}}^{\mathbf{n}} = TSD(\{ \mathcal{G}_{n_k}, \mathcal{C}_{n_k}\}^N_{k=1})$$
\label{theorem1}
\end{theorem}
\begin{proof} Let an order-$N$ tensor $\mathcal{X} \in \mathbb{R}^{I_1 \times I_2 \times \cdots \times I_N}$ have a Tensor Star decomposition of $TSD(\{ \mathcal{G}_k, \mathcal{C}_k\}^N_{k=1})$. Then, its element-wise representations can be written as 
\begin{equation}
    \begin{split}
    \mathcal{X}(i_1,i_2,...,i_N) =
    & \sum_{r_{1,1}=1}^{R_{1,1}}\sum_{r_{1,2}=1}^{R_{1,2}}\sum_{r_{2,1}=1}^{R_{2,1}}\sum_{r_{2,2}=1}^{R_{2,2}}\cdots \sum_{r_{N,1}=1}^{R_{N,1}}\sum_{r_{N,2}=1}^{R_{N,2}}\sum_{l_1=1}^{L_1}\sum_{l_2=1}^{L_2}\cdots \sum_{l_N=1}^{L_N} \\
     & \{ \mathcal{G}_1(r_{1,1},i_1,r_{1,2}) \ \mathcal{C}_1(r_{1,2}, l_1, l_{2}, r_{2,1}) \ \mathcal{G}_2(r_{2,1},i_2,r_{2,2}) 
 \ \mathcal{C}_2(r_{2,2}, l_2, l_{3}, r_{3,1}) \cdots \\
    & \ \ \mathcal{G}_k(r_{k,1},i_k,r_{k,2}) \ \mathcal{C}_k(r_{k,2}, l_k, l_{k+1}, r_{k,1}) \cdots \mathcal{G}_N(r_{N,1},i_N,r_{N,2}) \ \mathcal{C}_N(r_{N,2}, l_N, l_{1}, r_{1,1}) \}
    \end{split}
\end{equation}

Let $\mathbf{n} = (n_1, n_2, \ldots, n_N)$ be a circular dimensional permutation of $(1,2,\ldots,N)$. Then, the element-wise expression for $\vec{\mathcal{X}}^{\mathbf{n}}$ can be written as
\begin{equation}
    \begin{split}
    \vec{\mathcal{X}}^{\mathbf{n}}(i_{n_1},i_{n_2},\ldots ,i_{n_N}) =
    & \sum_{r_{1,1}=1}^{R_{1,1}}\sum_{r_{1,2}=1}^{R_{1,2}}\sum_{r_{2,1}=1}^{R_{2,1}}\sum_{r_{2,2}=1}^{R_{2,2}}\cdots\sum_{r_{N,1}=1}^{R_{N,1}}\sum_{r_{N,2}=1}^{R_{N,2}}\sum_{l_1=1}^{L_1}\sum_{l_2=1}^{L_2}\cdots \sum_{l_N=1}^{L_N} \\
     & \{\mathcal{G}_{n_1}(r_{{n_1},1},i_{n_1},r_{{n_1},2}) \ \mathcal{C}_{n_1}(r_{{n_1},2}, l_{n_1}, l_{{n_2}}, r_{{n_2},1}) \ \\
    & \ \ \mathcal{G}_{n_2}(r_{{n_2},1},i_{n_2},r_{{n_2},2}) 
 \ \mathcal{C}_2(r_{{n_2},2}, l_{n_2}, l_{{n_3}}, r_{{n_3},1}) \cdots \\
    & \ \ \mathcal{G}_{n_k}(r_{{n_k},1},i_{n_k},r_{{n_k},2}) \ \mathcal{C}_{n_k}(r_{{n_k},2}, l_{n_k}, l_{{n_{k+1}}}, r_{{n_k},1})\cdots \\
    & \ \ \mathcal{G}_{n_N}(r_{{n_k},1},i_{n_k},r_{{n_k},2}) \ \mathcal{C}_{n_k}(r_{{n_k},2}, l_{n_k}, l_{n_1}, r_{{n_1},1}) \}
    \end{split}
    \label{equ:proof_circular_invariance}
\end{equation}

Equation (\ref{equ:proof_circular_invariance}) can be shortened and written with tensor contractions:
\begin{equation}
    \begin{split}
    \vec{\mathcal{X}}^{\mathbf{n}} = &\mathcal{G}_{n_1} \times_3  \mathcal{C}_{n_1} \times_5 \mathcal{G}_{n_2} \times_{4,6}^{2,1}  \mathcal{C}_{n_2} \times_{6} \mathcal{G}_{n_3} \times_{5,7}^{2,1}  \mathcal{C}_{n_3} \times_{7} \cdots \\
    &\times_{k+3} \mathcal{G}_{n_k} \times_{k+2,k+4}^{2,1} \mathcal{C}_{n_k} \times_{k+4} \cdots \times_{N+3} \mathcal{G}_{n_N} \times_{1,3,N+2,N+4}^{4,3,2,1} \mathcal{C}_{n_N}
    \end{split}
\end{equation}

Therefore, the circularly permuted tensor $\vec{\mathcal{X}}^{\mathbf{n}}$ can be expressed as
$$\vec{\mathcal{X}}^{\mathbf{n}} = TSD(\{ \mathcal{G}_{n_k}, \mathcal{C}_{n_k}\}^N_{k=1})$$

\end{proof}

\begin{remark}
The TS decomposition allows for a direct correlation characterisation between any two dimensions of the original tensor $\mathcal{X}$. Consider the case where all core tensors $\{\mathcal{C}_k \}^N_{k=1}$ and one factor tensor $\mathcal{G}_k$ are "absorbed" into one tensor $\mathcal{B}$; then $\mathcal{B}$ becomes similar to the core tensor in FCTN decomposition. All other factor tensors $\{\mathcal{G}_i \}_{i=\{1,2,\ldots, N\}\setminus k} $ are all directly connected to $\mathcal{B}$ with double connections.
\end{remark}

\begin{theorem}[\textbf{Direct correlation characterisation between any two dimensions}]
Assume that an order-$N$ tensor $\mathcal{X} \in \mathbb{R}^{I_1 \times I_2 \times \cdots \times I_N}$ has a Tensor Star decomposition of $TSD(\{ \mathcal{G}_k, \mathcal{C}_k\}^N_{k=1})$ and let $\mathbf{n} = (n_1, n_2, \ldots, n_N)$ be a circular dimensional permutation of $(1,2,\ldots,N)$. Then, $\vec{\mathcal{X}}^{\mathbf{n}}$ can be rearranged into
\begin{equation}
    \begin{split}
        \vec{\mathcal{X}}^{\mathbf{n}} = & (\mathcal{G}_{n_1} \times^{2N,1}_{1,3} \mathcal{C}_{all} \times_{2k-2,2k-1}^{1,3} \mathcal{G}_{n_k} ) \times_{2,3}^{1,3} \mathcal{G}_{n_2} \times_{4,5}^{1,3} \mathcal{G}_{n_3} \times_{6,7}^{1,3} \mathcal{G}_{n_4} \times_{8,9}^{1,3} \\
        & \cdots \times_{k+1,k+2}^{1,3} \mathcal{G}_{n_{k-1}} \times_{k+3,k+4}^{1,3} \mathcal{G}_{n_{k+1}} \times_{k+5,k+6}^{1,3} \cdots \times_{2N-2,2N-1}^{1,3} \mathcal{G}_{n_{N}}
    \end{split}
\end{equation}
where $\mathcal{C}_{all} = \mathcal{C}_{n_1} \times_{3}^{2} \mathcal{C}_{n_2} \times_{5}^{2} \mathcal{C}_{n_3} \times_{7}^{2} \ldots \times_{2N-3}^{2} \mathcal{C}_{n_{N-1}} \times_{2,2N-1}^{3,2} \mathcal{C}_{n_N}$. Since $\mathcal{G}_{n_1}$ and $\mathcal{G}_{n_k}$ are directly connected through $\mathcal{C}_{all}$, any other two arbitrary dimensions of $\mathcal{X}$ can be directly connected through $\mathcal{C}_{all}$.
\label{theorem2}
\end{theorem}

\begin{proof} Assume that an order-$N$ tensor $\mathcal{X} \in \mathbb{R}^{I_1 \times I_2 \times \ldots \times I_N}$ has a Tensor Star decomposition $TSD(\{ \mathcal{G}_k, \mathcal{C}_k\}^N_{k=1})$. From Theorem \ref{theorem1}, if $\mathbf{n} = (n_1, n_2, \ldots, n_N)$ is a circular dimensional permutation of $(1,2,\cdots,N)$, then $\vec{\mathcal{X}}^{\mathbf{n}}$ can be written as
$$\vec{\mathcal{X}}^{\mathbf{n}} = TSD(\{ \mathcal{G}_{n_k}, \mathcal{C}_{n_k}\}^N_{k=1})$$

Since the TS decomposition core tensors are connected in a ring shape, tensor contraction can be applied to all $N$ order-$4$ core tensors to produce $\mathcal{C}_{all} \in \mathbb{R}^{R_{{n_1},2} \times R_{{n_2},1} \times R_{{n_2},2} \times R_{{n_3},1} \times \cdots  \times R_{{n_k},1} \times R_{{n_k},2} \times \cdots \times R_{{n_N},1} \times R_{{n_N},2} \times R_{{n_1},1}}$ as
\begin{equation}
\mathcal{C}_{all} = \mathcal{C}_{n_1} \times_{3}^{2} \mathcal{C}_{n_2} \times_{5}^{2} \mathcal{C}_{n_3} \times_{7}^{2} \cdots \times_{2N-3}^{2} \mathcal{C}_{n_{N-1}} \times_{2,2N-1}^{3,2} \mathcal{C}_{n_N} 
\end{equation}

It is now straightforward to show that $\mathcal{G}_{n_1} \in \mathbb{R}^{R_{n_1,1} \times I_{n_1} \times R_{n_1,2}}$ and $\mathcal{G}_{n_k} \in \mathbb{R}^{R_{n_k,1} \times I_{n_k} \times R_{n_k,2}}$ are connected through $\mathcal{C}_{all}$. This can be shown by obtaining $\mathcal{A} \in \mathbb{R}^{I_{n_1} \times R_{n_2,1} \times R_{n_2,2} \times R_{n_3,1} \times \cdots \times R_{n_N,1} \times R_{n_N,2} \times I_{n_k}} $ as
\begin{equation}
\mathcal{A} = \mathcal{G}_{n_1} \times^{2N,1}_{1,3} \mathcal{C}_{all} \times^{1,3}_{2k-2,2k-1} \mathcal{G}_{n_k} 
\end{equation}

The original $\vec{\mathcal{X}}^{\mathbf{n}}$ can then be obtained through combining all other $(N-2)$ factor tensors $\{ \mathcal{G}_{n_i}\}$,  where $i \in \{2, 3, \ldots, N\}\setminus k$, with $\mathcal{A}$, to yield
\begin{equation}
\vec{\mathcal{X}}^{\mathbf{n}} = (\mathcal{G}_{n_1} \times^{2N,1}_{1,3} \mathcal{C}_{all} \times^{1,3}_{2k-2,2k-1} \mathcal{G}_{n_k}) \times^{1,3}_{2,3} \mathcal{G}_{n_2} \times^{1,3}_{4,5} \mathcal{G}_{n_3} \times^{1,3}_{6,7} \mathcal{G}_{n_4} \times^{1,3}_{8,9} \cdots \times^{1,3}_{2N-2,2N-1} \mathcal{G}_{n_N}
\end{equation}

In this way, $\mathcal{G}_{n_1}$ and $\mathcal{G}_{n_k}$ have been proven to be directly connected through the core tensors. By varying the circular dimensional permutation vector $\mathbf{n}$ and $n_k$, it is obvious that any two arbitrary modes of $\mathcal{X}$ are directly connected through the core tensors.
\end{proof}

\subsection{Rank Bounds} When proposing a new tensor decomposition framework, it is essential to ensure that the low-rank property of the original tensor can be preserved. Thus, it is important to find the relation between the ranks of the original tensor and the latent ranks of TS decomposition. To this end, we next consider the rank bounds of TS decomposition. We note that TS decomposition provides more delicate control of the low-rankness than the FCTN or TW decompositions, as the Generalised Unfolding Rank Bound (GURB) of TS decomposition is always the product of $4$ terms and does not grow with the number of unfolded modes like in FCTN and TW decompositions.

\subsubsection{Mode-$n$ rank bound}

\begin{theorem}[\textbf{Mode-$n$ rank bound}]
Let an order-$N$ tensor $\mathcal{X} \in \mathbb{R}^{I_1 \times I_2 \times \cdots \times I_N}$ have a Tensor Star decomposition $TSD(\{ \mathcal{G}_k, \mathcal{C}_k\}^N_{k=1})$. For all $k=1,2,...,N$, the mode-$n$ rank bound of $\mathbf{X}_{(k)}$ and $\mathbf{X}_{<k>}$ can be expressed as
$$
\text{Rank}(\mathbf{X}_{(k)}) = \text{Rank}(\mathbf{X}_{<k>}) \leq R_{k,1}R_{k,2}, \ k=1,2,...,N
$$
\label{theorem3}
\end{theorem}
\begin{proof}
    Let an order-$N$ tensor $\mathcal{X} \in \mathbb{R}^{I_1 \times I_2 \times \cdots \times I_N}$ have a Tensor Star decomposition $TSD(\{ \mathcal{G}_k, \mathcal{C}_k\}^N_{k=1})$. Assume that $\vec{\mathcal{X}}^{\mathbf{n}}$ is a permuted $\mathcal{X}$, where $\mathbf{n} = (n_1, n_2, \ldots, n_N)$ is a circular dimensional permutation of $(1,2,\ldots,N)$. We first define a order-$(N+1)$ tensor $\mathcal{Y}_{all \setminus \mathcal{G}_{n_N}}$, obtained by contracting all other factor tensors $\{\mathcal{G}_{n_k}\}_{k=1}^{N-1}$and  core tensors $\{\mathcal{C}_{n_k}\}_{k=1}^{N}$, to give
\begin{equation}
    \begin{split}
    \mathcal{Y}_{all \setminus \mathcal{G}_{n_N}} = \ &\mathcal{G}_{n_1} \times_3  \mathcal{C}_{n_1} \times_5 \mathcal{G}_{n_2} \times_{4,6}^{2,1}  \mathcal{C}_{n_2} \times_{6} \mathcal{G}_{n_3} \times_{5,7}^{2,1}  \mathcal{C}_{n_3} \times_{7} \cdots  \\
    &\times_{k+3} \mathcal{G}_{n_k} \times_{k+2,k+4}^{2,1} \mathcal{C}_{n_k} \times_{k+4} \cdots \\
    &\times_{N+2} \mathcal{G}_{n_{N-1}} \times_{N+1,N+2}^{2,1} \mathcal{C}_{n_{N-1}} \times_{1,3,N+2}^{4,3,2} \mathcal{C}_{n_{N}}
    \end{split}
    \label{equ:y_to_isolate_g}
\end{equation}
where the order-$(N+1)$ tensor $\mathcal{Y}_{all \setminus \mathcal{G}_{n_N}} \in \mathbb{R}^{I_{n_1}\times I_{n_{2}}\times I_{n_3}\times \cdots \times I_{n_{N-2}}\times I_{n_{N-1}} \times R_{n_N,1} \times R_{n_N,2}}$.

Thus, $\vec{\mathcal{X}}^{\mathbf{n}}$ can be expressed as:
\begin{equation}
    \vec{\mathcal{X}}^{\mathbf{n}} = \mathcal{G}_{n_N} \times^{N,N+1}_{1,3} \mathcal{Y}_{all \setminus \mathcal{G}_{n_N}}
    \label{equ:isolate_g}
\end{equation}
To unfold the tensors in Equation (\ref{equ:isolate_g}), we first define a vector $\mathbf{y_v} = (N, N+1, 1, 2,\ldots, N-2, N-1)$ and obtain:
\begin{equation}
    \mathbf{X}_{<n_N>} = (\mathbf{G}_{n_N})_{(2)} (\mathbf{Y}_{all \setminus \mathcal{G}_{n_N}})_{[\mathbf{y_v};2]}
\end{equation}
where $(\mathbf{G}_{n_N})_{(2)} \in \mathbb{R}^{I_{n_N} \times R_{n_N,1} R_{n_N,2}}$ and $(\mathbf{Y}_{all \setminus \mathcal{G}_{n_N}})_{[\mathbf{y_v};2]} \in \mathbb{R}^{R_{n_N,1} R_{n_N,2} \times I_{n_1}  I_{n_2}  \cdots  I_{n_{N-2}}  I_{n_{N-1}}}$.

Since $\mathbf{n} = (n_1, n_2, \ldots, n_N)$ is a circular dimensional permutation of $(1,2,\ldots,N)$, for all $k=1,2,\ldots,N$, we can always let $n_N = k$, to arrive at:
\begin{equation}
    \mathbf{X}_{<k>} = (\mathbf{G}_{k})_{(2)} (\mathbf{Y}_{all \setminus \mathcal{G}_{k}})_{[\mathbf{y_v};2]}
    \label{equ:isolate_gk}
\end{equation}
Finally, this yields
\begin{equation}
\begin{split}
\text{Rank}(\mathbf{X}_{(k)}) = \text{Rank}(\mathbf{X}_{<k>}) & \leq \text{min}\{\text{Rank}((\mathbf{G}_{k})_{(2)}), \text{Rank}((\mathbf{Y}_{all \setminus \mathcal{G}_{k}})_{[\mathbf{y_v};2]}) \} \\
&\leq \text{Rank}((\mathbf{G}_{k})_{(2)}) \\
&\leq \text{min}\{I_k, R_{k,1}R_{k,2}\} \\
&\leq R_{k,1}R_{k,2}
\end{split}
\end{equation}
\end{proof}

\subsubsection{Generalised Unfolding Rank Bound}

Theorem \ref{theorem3} shows that the mode-$n$ rank is always bounded by the product of two factor ranks $R_{k,1}$ and $R_{k,2}$. Below we consider the rank bound of an unfolded $\mathcal{X}$, where its $d=2, 3, \ldots, N$ adjacent modes are unfolded into one dimension, i.e., $\mathbf{X}_{[1:N,d]} \in \mathbb{R}^{\prod^d_{i=1} I_i \times \prod^N_{j=d+1} I_j}$.

\begin{theorem}[\textbf{Generalised unfolding rank bound}]
Let an order-$N$ tensor $\mathcal{X} \in \mathbb{R}^{I_1 \times I_2 \times \cdots \times I_N}$ have a Tensor Star decomposition $TSD(\{ \mathcal{G}_k, \mathcal{C}_k\}^N_{k=1})$. Assume that $\vec{\mathcal{X}}^{\mathbf{n}}$ is a permuted $\mathcal{X}$, where $\mathbf{n} = (n_1, n_2, \ldots, n_N)$ is a circular dimensional permutation of $(1,2,\ldots,N)$. Then, for all $d=2,3,\ldots,N$, the generalised unfolding rank bound of $\mathbf{X}_{[n_1:n_N,d]}$ can be expressed as
$$
\text{Rank}(\mathbf{X}_{[n_1:n_N,d]}) \leq R_{n_{1},1}L_{n_1}L_{n_d} R_{n_{d},2}
$$
\label{theorem4}
\end{theorem}

\begin{proof}
Let an order-$N$ tensor $\mathcal{X} \in \mathbb{R}^{I_1 \times I_2 \times \cdots \times I_N}$ have a Tensor Star decomposition $TSD(\{ \mathcal{G}_k, \mathcal{C}_k\}^N_{k=1})$. Assume that $\vec{\mathcal{X}}^{\mathbf{n}}$ is a permuted $\mathcal{X}$, where $\mathbf{n} = (n_1, n_2, \ldots, n_N)$ is a circular dimensional permutation of $(1,2,\ldots,N)$. Thus, $\vec{\mathcal{X}}^{\mathbf{n}}$ has a TS decomposition of $TSD(\{ \mathcal{G}_{n_k}, \mathcal{C}_{n_k}\}^N_{k=1})$.

For any $d=2,3,\ldots,N$, we first define an order-$(d+4)$ tensor $\mathcal{P}$ and an order-$(N-d+4)$ tensor $\mathcal{Q}$ as follows: 
\begin{equation}
\mathcal{P} = \ \mathcal{G}_{n_1} \times_3  \mathcal{C}_{n_1} \times_5 \mathcal{G}_{n_2} \times_{4,6}^{2,1}  \mathcal{C}_{n_2} \times_{6} \mathcal{G}_{n_3} \times_{5,7}^{2,1}  \mathcal{C}_{n_3} \times_{7} \cdots \times_{d+2} \mathcal{G}_{n_{d-1}} \times_{d+1,d+3}^{2,1} \mathcal{C}_{n_{d-1}} \times_{d+3} \mathcal{G}_{n_d}
\end{equation}
\begin{equation}
\begin{split}
\mathcal{Q} = \ & \mathcal{C}_{n_d} \times_4  \mathcal{G}_{n_{d+1}}\times_{5,3}^{1,2} \mathcal{C}_{n_{d+1}} \times_{5}  \mathcal{G}_{n_{d+2}} \times_{6,4}^{1,2} \mathcal{C}_{n_{d+2}} \times_{6} \cdots \\ 
&\times_{(N-1)-d+3} \mathcal{G}_{n_{N-1}} \times_{(N-1)-d+4, (N-1)-d+2}^{1,2} \mathcal{C}_{n_{N-1}} \times_{N-d+3} \mathcal{G}_{n_N} \times^{1,2}_{N-d+4, N-d+2} \mathcal{C}_{n_N}
\end{split}
\end{equation}

where $\mathcal{P} \in \mathbb{R}^{R_{n_1,1} \times L_{n_1} \times I_{n_1} \times I_{n_2} \times I_{n_3} \times \cdots \times I_{n_{d-1}} \times L_{n_d} \times I_{n_d} \times R_{{n_d},2}}$ and 

$\mathcal{Q} \in \mathbb{R}^{R_{n_d,2} \times L_{n_d} \times I_{n_{d+1}} \times I_{n_{d+2}} \times I_{n_{d+3}} \times \cdots \times I_{n_{N-1}} \times I_{n_N} \times L_{n_1} \times R_{n_1,1} }$.

Then, $\mathcal{Q}$ and $\mathcal{P}$ can be contracted together to obtain $\vec{\mathcal{X}}^{\mathbf{n}}$ as
\begin{equation}
    \vec{\mathcal{X}}^{\mathbf{n}} = \mathcal{P} \times_{1,2,d+2,d+4}^{N-d+4,N-d+3,2,1} \mathcal{Q}
    \label{equ:X=QP}
\end{equation}
We now define two vectors $\mathbf{p_v}$ and $\mathbf{q_v}$ for the permutation of dimensions as
$$
\mathbf{p_v} = (3,4,5,\ldots,d-1,d,d+1,d+3,1,2,d+2,d+4)
$$
$$
\mathbf{q_v} = (N-d+4,N-d+3,2,1,3,4,5,6,\ldots,N-d-1,N-d,N-d+1,N-d+2)
$$
Upon unfolding $\mathcal{X}$, $\mathcal{P}$ and $\mathcal{Q}$, Equation (\ref{equ:X=QP}) can be written in a matrix format
\begin{equation}
    \mathbf{X}_{[n_1:n_N;d]} = \mathbf{P}_{[\mathbf{p_v};d]} \mathbf{Q}_{[\mathbf{q_v};4]}
\end{equation}

The rank bound of generalised unfolded $\mathbf{X}_{[1:N;d]}$ can now be expressed as
\begin{equation}
\begin{split}
    \text{Rank}(\mathbf{X}_{[n_1:n_N;d]})  & \leq \text{min}\{\text{Rank}(\mathbf{P}_{[\mathbf{p_v};d]}), \text{Rank}(\mathbf{Q}_{[\mathbf{q_v};4]}) \} \\
    &\leq \text{min}\{R_{n_1,1} L_{n_1} L_{n_d} R_{{n_d},2}, \prod^d_{i=1} I_{n_i}, \prod^N_{j=d+1} I_{n_j}\} \\
    &\leq R_{n_1,1} L_{n_1} L_{n_d} R_{n_d,2}
\end{split}
\end{equation}
where $d=2,3,...,N$.
\end{proof}

\section{Alternating Least Squares Algorithm for Tensor Star Decomposition}
An Alternating Least Squares (ALS) Algorithm is next developed for performing Tensor Star decomposition in this section. The TSD-ALS involves updating individual TS decomposition factor tensors, $\mathcal{G}_k$, and TS decomposition core tensors, $\mathcal{C}_k$, iteratively until the approximation error between $TSD(\{ \mathcal{G}_k, \mathcal{C}_k\}^N_{k=1})$ and the original order-$N$ tensor $\mathcal{X} \in \mathbb{R}^{I_1 \times I_2 \times I_3 \times \cdots \times I_{N-1} \times I_{N}}$ becomes lower than a set threshold $\epsilon$. We only derive the update steps for $\mathcal{G}_k$ and $\mathcal{C}_k$. As the subsequent ALS steps are straightforward and are therefore omitted.

\subsection{Objective Function of TSD-ALS}
The objective function of TSD-ALS is defined as follows:
\begin{equation}
    \min_{\{ \mathcal{G}_k, \mathcal{C}_k\}^N_{k=1}} || \mathcal{X} - TSD(\{ \mathcal{G}_k, \mathcal{C}_k\}^N_{k=1}) ||_F
\end{equation}

\subsection{Update $\mathcal{G}_{k}$}
To update a factor tensor $\mathcal{G}_k$ individually, it is essential to first isolate $\mathcal{G}_k$ from all other factor tensors and core tensors. Let $\mathbf{n} = (n_1, n_2, \ldots, n_N)$ be a circular dimensional permutation of $(1,2,\ldots,N)$, Equation (\ref{equ:y_to_isolate_g}) defines an order-$(N+1)$ tensor $\mathcal{Y}_{all \setminus \mathcal{G}_{n_N}}  \in \mathbb{R}^{I_{n_1}\times I_{n_{2}}\times I_{n_3}\times \cdots \times I_{n_{N-2}}\times I_{n_{N-1}} \times R_{n_N,1} \times R_{n_N,2}}$, which collects all TS decomposition factor and core tensors, except $\mathcal{G}_{n_N}$. Define a vector $\mathbf{y_v} = (N, N+1, 1, 2,\ldots, N-2, N-1)$. For $k=1,2,\ldots,N$, let $n_N=k$, we then have
\begin{equation}
    \mathbf{X}_{<k>} = (\mathbf{G}_{k})_{(2)} (\mathbf{Y}_{all \setminus \mathcal{G}_{k}})_{[\mathbf{y_v};2]}
    \label{equ:second-orderX=GY}
\end{equation}

Therefore, for $k=1,2,\ldots, N$, the factor tensor $\mathcal{G}_k$ can be updated by solving
\begin{equation}
\argmin_{(\mathbf{G}_{k})_{(2)}} ||\mathbf{X}_{<k>} - (\mathbf{G}_{k})_{(2)} (\mathbf{Y}_{all \setminus \mathcal{G}_{k}})_{[\mathbf{y_v};2]}||_F
\end{equation}

\subsection{Update $\mathcal{C}_{k}$}
Let $\mathbf{n} = (n_1, n_2, \ldots, n_N)$ be a circular dimensional permutation of $(1,2,\ldots,N)$. To update a core tensor $\mathcal{C}_k$ individually, we first define a tensor $\mathcal{Z}_{all \setminus \mathcal{C}_{n_N}} \in \mathbb{R}^{R_{n_1,1} \times I_{n_1} \times L_{n_1} \times I_{n_2} \times I_{n_3} \times \cdots \times I_{n_k} \times \cdots \times I_{n_{N-1}} \times L_{n_N} \times I_{n_N} \times R_{n_N,2}}$ to collect all TS decomposition factor and core tensors except $\mathcal{C}_{n_N}$, in the form
\begin{equation}
    \mathcal{Z}_{all \setminus \mathcal{C}_{n_N}} = \mathcal{G}_{n_1} \times_3  \mathcal{C}_{n_1} \times_5 \mathcal{G}_{n_2} \times_{4,6}^{2,1} \mathcal{C}_{n_2} \times_6 \mathcal{G}_{n_3} \times_{5,7}^{2,1} \mathcal{C}_{n_3} \times_7 \cdots \times_{k+3} \mathcal{G}_{n_k} \times_{k+2,k+4}^{2,1} \mathcal{C}_{n_k} \cdots \times_{N+3} \mathcal{G}_{n_N}
    \label{equ:z_def}
\end{equation}
In this way, we arrive at
\begin{equation}
   \vec{\mathcal{X}}^{\mathbf{n}} =  \mathcal{C}_{n_N} \times^{N+4,N+2,3,1}_{1,2,3,4} \mathcal{Z}_{all \setminus \mathcal{C}_{n_N}} 
   \label{equ:x=cz}
\end{equation}

Define the dimension permutation vector $\mathbf{z_v}=(N+4,N+2,3,1,2,4,5,6,\ldots, N-1, N, N+1, N+3)$. Then, Equation (\ref{equ:x=cz}) can be written in vector and matrix format as
\begin{equation}
    \mathbf{x}_{[n_1:n_N; 0]} = (\mathbf{c}_{n_N})_{[1:4; 0]} (\mathbf{Z}_{all \setminus \mathcal{C}_{n_N}})_{[\mathbf{z_v};4]}
\end{equation}

For $k=1,2,\ldots,N$, let $n_N=k$, to have
\begin{equation}
    \mathbf{x}_{[n_1:n_N; 0]} = (\mathbf{c}_{k})_{[1:4; 0]} (\mathbf{Z}_{all \setminus \mathcal{C}_{k}})_{[\mathbf{z_v};4]}
\label{equ:second-orderX=CZ}
\end{equation}

Therefore, for $k=1,2,\ldots, N$, the core tensor $\mathcal{C}_k$ can be updated by solving
\begin{equation}
    \argmin_{(\mathbf{c}_{k})_{[1:4; 0]}} || \mathbf{x}_{[n_1:n_N; 0]} - (\mathbf{c}_{k})_{[1:4; 0]} (\mathbf{Z}_{all \setminus \mathcal{C}_{k}})_{[\mathbf{z_v};4]}||_F
\end{equation}

\section{Proximal Alternating Minimization for Tensor Completion}
Tensor decompositions can be used to promote sparsity in the data and find efficient representations of the data. In this way, by assuming a low-rank structure in some data with missing or untrustworthy entries, we can update the tensor decomposed components iteratively to approximate the observed "good" entries while filling up the missing "bad" entries, the so-called tensor completion. To this end, we shall now introduce a Proximal Alternating Minimization (PAM) algorithm for the proposed Tensor Star decomposition similar to \cite{zheng2021fully} for FCTN and \cite{twd} for TW decomposition. The TSD-PAM algorithm is shown in Algorithm \ref{alg:TSD-PAM}. Note that the TSD-PAM algorithm can also be used to obtain the TS decomposition of a tensor if there are no missing entries.

As only the observed entries are approximated, let $\delta_{\Omega}()$ denote the indicator function of whether a location belongs to the location set $\Omega$ of observed entries in a tensor $\mathcal{H} \in \mathbb{R}^{I_1 \times I_2 \times \cdots \times I_N}$. Let $\mathcal{X} \in \mathbb{R}^{I_1 \times I_2 \times \cdots \times I_N} $ be an approximation of $\mathcal{H}$. According to \cite{bolte2014proximal} and \cite{attouch2013convergence}, the unconstrained problem can be formulated as
\begin{equation}
    \min_{\mathcal{X}, \{ \mathcal{G}_i, \mathcal{C}_i\}_{i=1}^{N}} \psi(\mathcal{X}, \{ \mathcal{G}_i, \mathcal{C}_i\}_{i=1}^{N}) = \min_{\mathcal{X}, \{ \mathcal{G}_i, \mathcal{C}_i\}_{i=1}^{N}}\frac{1}{2}||\mathcal{X}-TSD(\{ \mathcal{G}_k, \mathcal{C}_k\}^N_{k=1}) ||^2_F + \delta_{\Omega}(\mathcal{X})
    \label{equ:objective_f}
\end{equation}
where\begin{align}
\delta_{\Omega}(x_{i_1 i_2 \cdots i_N})=
\left\{
\begin{aligned}
    0, \ \ &\{i_1 i_2 i_3 \cdots i_N\} \in \Omega \\
    \mathcal{1}, \ \ &\{i_1 i_2 i_3 \cdots i_N\} \notin \Omega
\end{aligned}
\right.
\end{align}

In this way, we can derive the updating steps by solving
\begin{align}
\left\{
\begin{aligned}
    \mathcal{G}^{(t+1)}_{k} & = \argmin_{\mathcal{G}_{k}}\{\psi(\mathcal{X}^{(t)}, \{\mathcal{G}^{(t+1)}_i, \mathcal{C}^{(t+1)}_i\}_{i=1}^{k-1}, \mathcal{G}_k, \mathcal{C}^{(t)}_k , \{\mathcal{G}^{(t)}_i, \mathcal{C}^{(t)}_i\}_{i=k+1}^{N}) +\frac{\rho}{2}||\mathcal{G}_{k} - \mathcal{G}^{(t)}_{k} ||^2_F \} \\
    \mathcal{C}^{(t+1)}_{k} & = \argmin_{\mathcal{C}_{k}}\{\psi(\mathcal{X}^{(t)}, \{\mathcal{G}^{(t+1)}_i, \mathcal{C}^{(t+1)}_i\}_{i=1}^{k-1}, \mathcal{G}^{(t+1)}_k, \mathcal{C}_k , \{\mathcal{G}^{(t)}_i, \mathcal{C}^{(t)}_i\}_{i=k+1}^{N})+\frac{\rho}{2}||\mathcal{C}_{k} - \mathcal{C}^{(t)}_{k} ||^2_F \} \\
    \mathcal{X}^{(t+1)}_{k} & = \argmin_{\mathcal{X}_{k}}\{\psi(\mathcal{X}, \{ \mathcal{G}^{(t+1)}_i, \mathcal{C}^{(t+1)}_i\}_{i=1}^{N}) +\frac{\rho}{2}||\mathcal{X} - \mathcal{X}^{(t)} ||^2_F \}
\end{aligned}
\right.
\label{equ:pam_problems}
\end{align}
where $\rho$ is a positive proximal parameter, and $t \in \mathbb{Z}^{+}$ denotes the iteration.

\begin{algorithm}
    \caption{Proximal Alternating Minimization for Tensor Completion using Tensor Star Decomposition}
    \label{alg:TSD-PAM}
    \begin{algorithmic}[1] 
        \State \textbf{Input:} Partially observed $\mathcal{H}$ within the set $\Omega$, factor ranks $\{R_{i,s}\}_{i=1, s=1}^{N,2}$, ring ranks $\{L\}_{i=1}^{N}$, initialised $\mathcal{X}^{(1)} = \mathcal{H}$ within the set $\Omega$ and $0$ otherwise, initialised factor tensors $\{\mathcal{G}^{(1)}_{k}\}_{k=1}^{N}$ and core tensors $\{\mathcal{C}^{(1)}_{k}\}_{k=1}^{N}$ from $\mathcal{N}(0,1)$, $\rho=0.01$, maximum iteration $itr=1000$, and convergence threshold $tol=0.00001$.
        \For{$t=1,2,\ldots,itr$}
        \For{$k=1,2,\ldots,N$}
            \State Update $\mathcal{G}_{k}$ according to Equation (\ref{equ:g_update}).
            \State Update $\mathcal{C}_{k}$ according to Equation (\ref{equ:c_update}).
        \EndFor
        \State Update $\mathcal{X}$ according to Equation (\ref{equ:x_update}).
        \If $|| \mathcal{X}^{(t)}-\mathcal{X}^{(t-1)}||_F / ||\mathcal{X}^{(t)} < tol$
        \State break
        \EndIf
        \EndFor

    \State \textbf{Output:} $\{\mathcal{G}_{k}\}_{k=1}^{N}$, $\{\mathcal{C}_{k}\}_{k=1}^{N}$, and $\mathcal{X}$ which represents the recovered $\mathcal{H}$.
    \end{algorithmic}
\end{algorithm}

\subsection{Update of the factor tensor $\mathcal{G}_{k}$} \label{sec:update_g}
Using Equation (\ref{equ:second-orderX=GY}), the $\{\mathcal{G}_{i}\}_{i=1}^{N}$ subproblems of Equation (\ref{equ:pam_problems}) can be written in a order-$2$ format as
\begin{equation}
    (\mathbf{G}_{k}^{(t+1)})_{(2)} = \argmin_{(\mathbf{G}_{k})_{(2)}}  \{\frac{1}{2}||\mathbf{X}_{<k>} - (\mathbf{G}_{k})_{(2)} (\mathbf{Y}_{all \setminus \mathcal{G}_{k}})_{[\mathbf{y_v};2]} ||^2_F + \frac{\rho}{2}||(\mathbf{G}_{k})_{(2)} - (\mathbf{G}_{k}^{(t)})_{(2)} ||^2_F\}
    \label{equ:g_sub}
\end{equation}
where is $\mathcal{Y}_{all \setminus \mathcal{G}_{k}}$ is defined in Equation (\ref{equ:y_to_isolate_g}) and is a function of $(\{\mathcal{G}^{(t+1)}_i, \mathcal{C}^{(t+1)}_i\}_{i=1}^{k-1}, \mathcal{C}^{(t)}_k , \{\mathcal{G}^{(t)}_i, \mathcal{C}^{(t)}_i\}_{i=k+1}^{N})$.

Upon setting the first derivative of Equation (\ref{equ:g_sub}) w.r.t. $(\mathbf{G}_{k})_{(2)}$ to $0$, we have:
\begin{equation}
        -(\mathbf{X}_{<k>}-(\mathbf{G}_{k})_{(2)} (\mathbf{Y}_{all \setminus \mathcal{G}_{k}})_{[\mathbf{y_v};2]}) ((\mathbf{Y}_{all \setminus \mathcal{G}_{k}})_{[\mathbf{y_v};2]})^{\mathsf{T}} + \rho ((\mathbf{G}_{k})_{(2)} - (\mathbf{G}_{k}^{(t)})_{(2)}) = 0
\end{equation}
\begin{equation}
       (\mathbf{G}_{k})_{(2)}(\rho \mathbf{I} + (\mathbf{Y}_{all \setminus \mathcal{G}_{k}})_{[\mathbf{y_v};2]}((\mathbf{Y}_{all \setminus \mathcal{G}_{k}})_{[\mathbf{y_v};2]})^{\mathsf{T}} )=\rho (\mathbf{G}_{k}^{(t)})_{(2)} + \mathbf{X}_{<k>}((\mathbf{Y}_{all \setminus \mathcal{G}_{k}})_{[\mathbf{y_v};2]})^{\mathsf{T}} 
\end{equation}
Then,
\begin{equation}
       (\mathbf{G}^{(t+1)}_{k})_{(2)}=(\rho (\mathbf{G}_{k})_{(2)} + \mathbf{X}_{<k>}((\mathbf{Y}_{all \setminus \mathcal{G}_{k}})_{[\mathbf{y_v};2]})^{\mathsf{T}}) (\rho \mathbf{I} + (\mathbf{Y}_{all \setminus \mathcal{G}_{k}})_{[\mathbf{y_v};2]}((\mathbf{Y}_{all \setminus \mathcal{G}_{k}})_{[\mathbf{y_v};2]})^{\mathsf{T}} )^{-1}
       \label{equ:g_update}
\end{equation}

From here, it is straightforward to see that folding $ (\mathbf{G}_{k}^{(t+1)})_{(2)}$ gives $\mathcal{G}^{(t+1)}_{k} $.

\subsection{Update of the core tensor $\mathcal{C}_{k}$}  \label{sec:update_c}
Using Equation (\ref{equ:second-orderX=CZ}), the $\{\mathcal{C}_{i}\}_{i=1}^{N}$ subproblems of Equation (\ref{equ:pam_problems}) can be written as

\begin{equation}
    (\mathbf{c}^{(t+1)}_{k})_{[1:4; 0]} = \argmin_{(\mathbf{c}_{k})_{[1:4; 0]}}  \{\frac{1}{2}||\mathbf{x}_{[n_1:n_N; 0]}^{(t)} - (\mathbf{c}_{k})_{[1:4; 0]} (\mathbf{Z}_{all \setminus \mathcal{C}_{k}})_{[\mathbf{z_v};4]} ||^2_F + \frac{\rho}{2}||(\mathbf{c}_{k})_{[1:4; 0]} - (\mathbf{c}^{(t+1)}_{k})_{[1:4; 0]} ||^2_F\}
    \label{equ:c_sub}
\end{equation}
where is $\mathcal{Z}_{all \setminus \mathcal{C}_{k}}$ is defined in Equation (\ref{equ:z_def}) and is a function of $(\{\mathcal{G}^{(t+1)}_i, \mathcal{C}^{(t+1)}_i\}_{i=1}^{k-1}, \mathcal{G}^{(t+1)}_k, \{\mathcal{G}^{(t)}_i, \mathcal{C}^{(t)}_i\}_{i=k+1}^{N})$.

Upon setting the first derivative of Equation (\ref{equ:c_sub}) w.r.t. $(\mathbf{c}_{k})_{[1:4; 0]}$ to $0$, we have:
\begin{equation}
-(\mathbf{x}_{[n_1:n_N; 0]}^{(t)}-(\mathbf{c}^{(t+1)}_{k})_{[1:4; 0]} (\mathbf{Z}_{all \setminus \mathcal{C}_{k}})_{[\mathbf{z_v};4]})((\mathbf{Z}_{all \setminus \mathcal{C}_{k}})_{[\mathbf{z_v};4]})^{\mathsf{T}} + \rho((\mathbf{c}^{(t+1)}_{k})_{[1:4; 0]}-(\mathbf{c}^{(t)}_{k})_{[1:4; 0]}) = 0
\end{equation}
\begin{equation}
(\mathbf{c}^{(t+1)}_{k})_{[1:4; 0]}(\rho \mathbf{I} + (\mathbf{Z}_{all \setminus \mathcal{C}_{k}})_{[\mathbf{z_v};4]}((\mathbf{Z}_{all \setminus \mathcal{C}_{k}})_{[\mathbf{z_v};4]})^{\mathsf{T}}) =    \rho(\mathbf{c}^{(t)}_{k})_{[1:4; 0]}+\mathbf{x}_{[n_1:n_N; 0]}^{(t)}((\mathbf{Z}_{all \setminus \mathcal{C}_{k}})_{[\mathbf{z_v};4]})^{\mathsf{T}}
\end{equation}
Then,
\begin{equation}
    (\mathbf{c}^{(t+1)}_{k})_{[1:4; 0]} = (\rho(\mathbf{c}^{(t)}_{k})_{[1:4; 0]}+\mathbf{x}_{[n_1:n_N; 0]}^{(t)}((\mathbf{Z}_{all \setminus \mathcal{C}_{k}})_{[\mathbf{z_v};4]})^{\mathsf{T}}) ((\rho \mathbf{I} + (\mathbf{Z}_{all \setminus \mathcal{C}_{k}})_{[\mathbf{z_v};4]}((\mathbf{Z}_{all \setminus \mathcal{C}_{k}})_{[\mathbf{z_v};4]})^{\mathsf{T}}))^{-1}
    \label{equ:c_update}
\end{equation}

From here, it is straightforward to see that folding $ (\mathbf{c}^{(t+1)}_{k})_{[1:4; 0]}$ gives $\mathcal{C}_{k}^{(t+1)}$.

\subsection{Update of $\mathcal{X}$}  \label{sec:update_x}
The $\mathcal{X}$-subproblem of Equation (\ref{equ:pam_problems}) can be expressed as
\begin{equation}
    \mathcal{X}^{(t+1)}=\argmin_{\mathcal{X}}\{ \frac{1}{2}||\mathcal{X}-TSD(\{ \mathcal{G}^{(t+1)}_k, \mathcal{C}^{(t+1)}_k\}^N_{k=1}) ||^2_F + \delta_{\Omega}(\mathcal{X})+\frac{\rho}{2}||\mathcal{X} - \mathcal{X}^{(t)} ||^2_F\}
    \label{equ:x_subproblem}
\end{equation}

The indicator function $\delta_{\Omega}(\mathcal{X})$ in Equation (\ref{equ:x_subproblem}) essentially constrains the reduction of the approximation error to only within the location set $\Omega$, while forcing the entries of $\mathcal{X}^{(t+1)}$ in location set $\Omega$ to be equal to the entries of $\mathcal{H}$ in location set $\Omega$. Thus, we can rewrite the $\mathcal{X}$-subproblem as
\begin{align}
\mathcal{X}^{(t+1)}(i_1, i_2, \ldots, i_N) =
\left\{
\begin{aligned}
&\mathcal{H}(i_1, i_2, \ldots, i_N), i_1 i_2 \cdots i_N  \in \Omega \\
&\argmin_{x}\{ \frac{1}{2}||\mathcal{X}-TSD(\{ \mathcal{G}^{(t+1)}_k, \mathcal{C}^{(t+1)}_k\}^N_{k=1}) ||^2_F + \frac{\rho}{2}||\mathcal{X} - \mathcal{X}^{(t)} ||^2_F\},  i_1 i_2 \cdots i_N  \notin \Omega
\end{aligned}
\right.
\label{equ:x_sub}
\end{align}

To solve the second part of Equation (\ref{equ:x_sub}), we take the first derivative of the second part of Equation (\ref{equ:x_sub}) w.r.t. $\mathcal{X}$ and set it to $0$, to give
\begin{equation}
\mathcal{X}-TSD(\{ \mathcal{G}^{(t+1)}_k, \mathcal{C}^{(t+1)}_k\}^N_{k=1}) + \rho \mathcal{X} - \rho \mathcal{X}^{(t)} = 0
\end{equation}
\begin{equation}
(1+\rho)\mathcal{X} = TSD(\{ \mathcal{G}^{(t+1)}_k, \mathcal{C}^{(t+1)}_k\}^N_{k=1}) + \rho \mathcal{X}^{(t)}
\end{equation}
Therefore, 
\begin{equation}
x^{(t+1)}_{i_1 i_2 \cdots i_N} = \frac{TSD(\{ \mathcal{G}^{(t+1)}_k, \mathcal{C}^{(t+1)}_k\}^N_{k=1}) + \rho \mathcal{X}^{(t)}}{1+\rho}, i_1 i_2 \cdots i_N  \notin \Omega
\end{equation}

Finally, we arrive at the full update step of $\mathcal{X}^{(t)}$ in the form
\begin{align}
x^{(t+1)}_{i_1 i_2 \cdots i_N} =
\left\{
\begin{aligned}
&\mathcal{H}(i_1, i_2, \cdots, i_N), i_1 i_2 \cdots i_N  \in \Omega \\
&\frac{TSD(\{ \mathcal{G}^{(t+1)}_k, \mathcal{C}^{(t+1)}_k\}^N_{k=1}) + \rho \mathcal{X}^{(t)}}{1+\rho}, i_1 i_2 \cdots i_N  \notin \Omega
\end{aligned}
\right.
\label{equ:x_update}
\end{align}

\section{Proof-of-concept Example of TSD-PAM}
We next demonstrate a proof-of-concept example of TSD-PAM to verify the performance and convergence of the algorithm in practice by employing TSD-PAM to popular datasets of tensor completion. More specifically, we utilise the same test MSI data, $Toy$, of size $200\times200\times31$ \footnote{https://www.cs.columbia.edu/CAVE/databases/multispectral/} as in \cite{twd}. With $90 \%$ missing entries and an observed MPSNR of 11.389, TSD-PAM with small ranks $L_1=3$, $R_{3,1} = 4$, $R_{3,2}=L_2=6$, $R_{1,1}= L_3= 8$, $R_{1,2} = R_{2,1} = R_{2,2} =9$, achieves an MPSNR of $36.5466$, as shown in Figure \ref{fig:toy_result}. 

\begin{figure}[h] 
    \centering
    \includegraphics[width = 0.9\hsize]{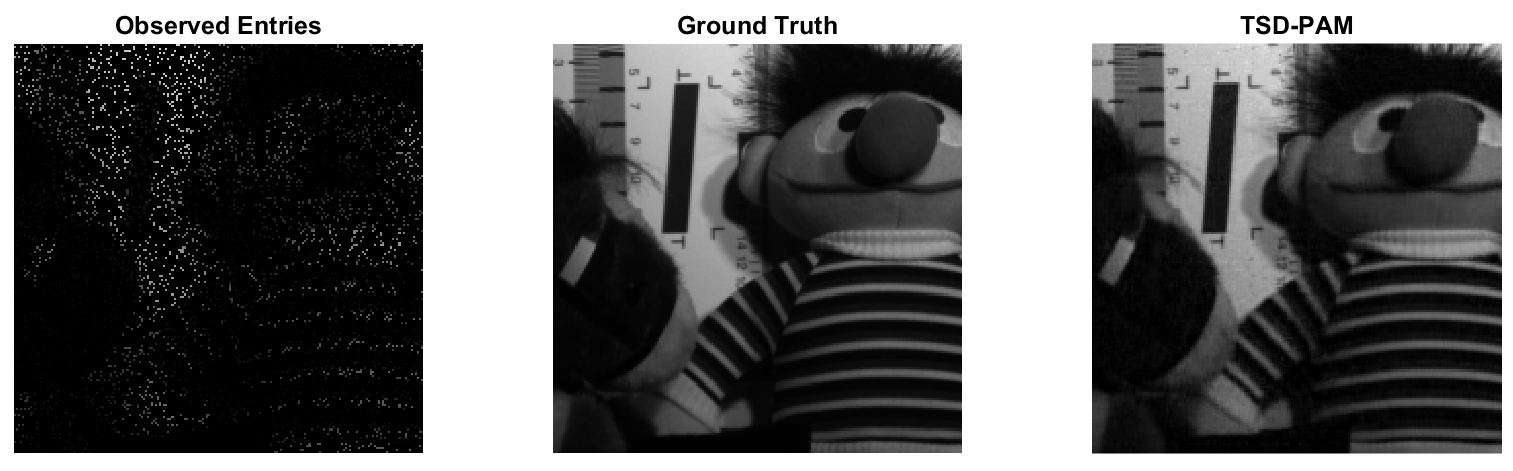}
    \caption{With $90\%$ entries missing, TSD-PAM is able to recover the original tensor with small ranks (The largest rank is $9$). The 20-th spectral sample is shown.}
    \label{fig:toy_result}
\end{figure}

\pagebreak
\bibliographystyle{ieeetr}  
\bibliography{references}  
\end{document}